\documentclass[11pt]{article}
\usepackage[hmargin=0.9in,vmargin=1in]{geometry}
\usepackage[T1]{fontenc}
\usepackage[utf8]{inputenc}
\usepackage{lmodern,microtype,mathtools,amsthm,amssymb,enumitem,xcolor,hyperref,bookmark}

\hypersetup{
    unicode,
    colorlinks=true,
    citecolor=green!40!black,
    linkcolor=red!20!black,
    urlcolor=blue!40!black,
    filecolor=cyan!30!black
    }

\usepackage[capitalise,nameinlink,sort]{cleveref}
\crefname{claim}{claim}{claims}
\Crefname{claim}{Claim}{Claims}

\usepackage{thm-restate}
\usepackage{cite}

\setlist[itemize]{leftmargin=2em,labelsep=.6em,topsep=\medskipamount,noitemsep}
\setlist[enumerate]{leftmargin=2em,labelsep=.5em,topsep=\medskipamount,noitemsep}

\newtheorem{theorem}{Theorem}[section]
\newtheorem{lemma}[theorem]{Lemma}
\newtheorem{corollary}[theorem]{Corollary}

\newtheorem{question}[theorem]{Question}
\newtheorem{claim}[theorem]{Claim}

\makeatletter
\newdimen\@savedtopsep
\renewenvironment{proof}[1][\proofname]{\par\@savedtopsep\topsep
\pushQED{\qed}\normalfont\topsep6\p@\@plus6\p@\trivlist
\item[\hskip\labelsep\itshape #1\@addpunct{.}]\topsep\@savedtopsep
\ignorespaces}{\popQED\endtrivlist\@endpefalse}
\makeatother

\crefname{problem}{Problem}{Problems}
\crefname{claim}{Claim}{Claims}

\newcommand{\Oh}{\mathcal{O}}

\newcommand{\cT}{\mathcal{T}}
\newcommand{\C}{\mathsf{C}}
\newcommand{\K}{\mathsf{K}}
\newcommand{\M}{\mathsf{M}}
\newcommand{\N}{\mathsf{N}}

\DeclareMathOperator{\yolov}{tree\textnormal{-}\mu}
\DeclareMathOperator{\treealpha}{tree\textnormal{-}\alpha}

\DeclarePairedDelimiter{\size}{\lvert}{\rvert}

\let\leq\leqslant\let\le\leq
\let\geq\geqslant\let\ge\geq
\let\setminus\smallsetminus

\usepackage[textsize=footnotesize, textwidth=2cm, color=yellow]{todonotes}

\title{Induced matching treewidth and tree-independence number, revisited}

\author{Noga Alon\thanks{ Princeton University, Princeton, NJ, USA (\texttt{nalon@math.princeton.edu}).
Supported in part by NSF grant DMS-2154082.}
\and
Martin Milanič\thanks{FAMNIT and IAM, University of Primorska, Slovenia (\texttt{martin.milanic@upr.si}).
Supported in part by the Slovenian Research and Innovation Agency (I0-0035, research program P1-0285 and research projects J1-3003, J1-4008, J1-4084, J1-60012, and N1-0370) and by the research program CogniCom (0013103) at the University of Primorska.}
\and Paweł Rzążewski\thanks{Warsaw University of Technology \& University of Warsaw, Poland (\texttt{pawel.rzazewski@pw.edu.pl}).
Supported by the National Science Centre grant number 2024/54/E/ST6/00094.}}

\begin{document}

\maketitle

\begin{abstract}
We study two graph parameters defined via tree decompositions: tree-independence number and induced matching treewidth.
Both parameters are defined similarly as treewidth, but with respect to different measures of a tree decomposition $\mathcal{T}$ of a graph $G$: for tree-independence number, the measure is the maximum size of an independent set in $G$ included in some bag of $\mathcal{T}$, while for the induced matching treewidth, the measure is the maximum size of an induced matching in $G$ such that some bag of $\mathcal{T}$ contains at least one endpoint of every edge of the matching.

While the induced matching treewidth of any graph is bounded from above by its tree-independence number, the family of complete bipartite graphs shows that small induced matching treewidth does not imply small tree-independence number.
On the other hand, Abrishami, Bria\'nski, Czy\.zewska, McCarty, Milani\v{c}, Rz\k{a}\.zewski, and Walczak~[SIAM Journal on Discrete Mathematics, 2025] showed that, if a fixed biclique $K_{t,t}$ is excluded as an induced subgraph, then the tree-independence number is bounded from above by some function of the induced matching treewidth.
The function resulting from their proof is exponential even for fixed $t$, as it relies on multiple applications of Ramsey's theorem.
In this note we show, using the K\"ov\'ari-S\'os-Tur\'an theorem, that for any class of $K_{t,t}$-free graphs, the two parameters are in fact polynomially related.
\end{abstract}

\section{Introduction}

Treewidth is a graph parameter that, roughly speaking, measures how similar the graph is to a tree.
The notion of treewidth was introduced independently several times and in different contexts (see~\cite{ACP87,BB72,Halin76,RS84}) and has been an important tool in graph theory, for both structural as  well as algorithmic reasons (see, e.g.,~\cite{MR2188176} and~\cite{MR1105479,MR1042649}, respectively).
The definition is based on the notion of a \textsl{tree decomposition} of a graph, that is, a collection of subsets of the vertex set of the graph called \emph{bags} such that the endpoints of each edge of the graph appear in some bag, and, moreover, the bags are arranged in a tree so that the bags containing any fixed vertex of the graph form a nonempty subtree (see \cref{sec:prelim} for a precise definition).
In the case of treewidth, the aim is to find a tree decomposition that minimizes the maximum size of a bag.

One of the limitations of treewidth is that graphs with bounded treewidth are necessarily sparse.
To overcome this limitation, several more general graph parameters have been defined in the literature that can also be bounded on dense graphs, while still retaining some of the good features of treewidth (see, e.g.,~\cite{hlinveny2008width,MR3721445,MR4402362,Yolov18}).
In particular, Yolov~\cite{Yolov18} and independently Dallard, Milanič, and Štorgel~\cite{DMS24a} introduced the notion of \emph{tree-independence number} (denoted by $\treealpha$), which is a graph parameter based on tree decompositions that is defined similarly as treewidth, but where instead of bounding the maximum size of a bag in a tree decomposition, what matters is the maximum size of an \textsl{independent set} contained in a bag.
This parameter properly generalizes treewidth, as it can indeed capture dense graphs.
For example, chordal graphs are known to admit tree decompositions such that every bag is a clique (see, e.g.,~\cite{MR1320296}); in other words, they have tree-independence number $1$.
Classes of graphs with bounded tree-independence number have interesting structural properties (for example, treewidth of such graphs can only be large due to the presence of a large clique, see~\cite{DMS24a}) and admit polynomial-time algorithms for several problems that are \textsf{NP}-hard on general graphs (see~\cite{DFGKM25,DMS24a,LMMORS-ESA24,Yolov18,MR4640320}).

However, some highly structured graphs, such as balanced complete bipartite graphs $K_{t,t}$, can have large treewidth even though they do not have large cliques; consequently, such graphs do not have bounded tree-independence number.
This observation motivated Yolov~\cite{Yolov18} to introduce and study a further generalization of tree-independence number, called \emph{induced matching treewidth}\footnote{Actually, Yolov defined this parameter for general hypergraphs and called it \emph{minor-matching hypertreewidth}. The terminology and notation we use comes from the work of Lima et al.~\cite{LMMORS-ESA24}, as it is more suitable to our setting.} (and denoted by $\yolov$), which is also based on tree decompositions, but the measure is the maximum size of an induced matching in the graph such that some bag of the tree decomposition intersects every edge of the matching.
Induced matching treewidth generalizes tree-independence number, in the sense that bounded tree-independence number implies bounded induced matching treewidth.
The generalization is proper: complete bipartite graphs, which have  induced matching treewidth~$1$ but arbitrarily large tree-independence number, show that tree-independence number cannot be bounded from above by any function of induced matching treewidth.

As shown already by Yolov~\cite{Yolov18}, classes of graphs with bounded induced matching treewidth enjoy some of the good algorithmic properties of graph classes of bounded tree-independence number. 
In particular, Lima, Milanič, Muršič, Okrasa, Rzążewski, and Štorgel (see~\cite{LMMORS-ESA24,LMMORS24}) conjectured that classes with bounded induced matching treewidth admit a polynomial-time algorithm for a meta-problem defined by a fixed \textsf{CMSO}$_2$-sentence $\Phi$ and a fixed integer $k$, where for a given vertex-weighted graph $G$, the task is to find a maximum-weight set $X\subseteq V(G)$ that induces a subgraph with treewidth at most $k$ that satisfies $\Phi$.
This conjecture was recently proved by Bodlaender, Fomin, and Korhonen~\cite{bodlaender2025finding}.

From the structural point of view, both tree-independence number and induced matching treewidth have been studied in recent literature, see~\cite{MR4906164,LMMORS24} for induced matching treewidth and~\cite{chudnovsky2024treeindependencenumberivsoda,chudnovsky2024treeindependencenumberiii,MR4955546,abrishami2024tree,chudnovsky2025treeindependencenumberv,dallard2024treewidthversuscliquenumber,DMS24b,DBLP:conf/soda/AhnGHK25,hilaire2025treewidthversuscliquenumber,MR4970654} for tree-independence number.
In particular, the two parameters were studied in relation to each other and to other width parameters, see~\cite{bergougnoux2023newwidthparametersindependent,MR4906164,LMMORS24}.
For instance, Lima et al.~\cite{LMMORS24} compared the two parameters in the context of (distance) powers of graphs, showing among other things that ${\treealpha(G^{k})\le \yolov(G)}$ for every odd integer $k\ge 3$ and every graph $G$ with at least one edge.
Furthermore, Abrishami, Bria\'nski, Czy\.zewska, McCarty, Milani\v{c}, Rz\k{a}\.zewski, and Walczak~\cite{MR4906164} proved that for graph classes closed under induced subgraphs, complete bipartite graphs are the only reason why bounded induced matching treewidth does not imply bounded tree-independence number.

\begin{theorem}[Abrishami et al.]\label{thm:bicliqueold}
For every two positive integers\/ $\mu$ and\/ $t$, there is an integer\/ $\K(\mu,t)$ such that the following holds.
Every\/ $K_{t,t}$-free graph\/ $G$ with\/ $\yolov(G)\leq\mu$ satisfies\/ $\treealpha(G)<\K(\mu,t)$.
\end{theorem}

A systematic comparison of various graph width parameters under the assumption of excluding some $K_{t,t}$ as a subgraph or induced subgraph was done recently by Brettell, Munaro, Paulusma, and Yang (see~\cite{MR4901497}).
Their work leaves open the following question: Is it true that in the absence of some fixed complete bipartite graph as an induced subgraph, bounded \textsl{sim-width} implies bounded tree-independence number? (For the definition of sim-width, see, e.g., the paper of Brettell et al.~\cite{MR4901497}.)
An affirmative answer would generalize \Cref{thm:bicliqueold}.

The purpose of this note is to give a different, quantitative improvement of \Cref{thm:bicliqueold}.
The function resulting from the proof of \Cref{thm:bicliqueold} in~\cite{MR4906164} relies on multiple applications of Ramsey's theorem and is thus exponential in $\mu$ even for fixed $t$.
We show, using the K\"ov\'ari-S\'os-Tur\'an theorem~\cite{MR65617}, that for any class of $K_{t,t}$-free graphs, induced matching treewidth and tree-independence number are in fact \textsl{polynomially} related.

\begin{restatable}{theorem}{thmbiclique}
\label{thm:biclique}
For every two positive integers\/ $\mu$ and\/ $t$, every\/ $K_{t,t}$-free graph\/ $G$ with\/ $\yolov(G)\leq\mu$ satisfies\/ $\treealpha(G)= \Oh_t(\mu^{3t^2+1})$.
\end{restatable}

After giving the necessary definitions in \Cref{sec:prelim}, we prove in
\Cref{sec:auxiliary-lemmas} two auxiliary Ramsey-type results with polynomial bounds.
Our main result, \cref{thm:biclique}, is proved in \Cref{sec:proof}.
In conclusion, we pose some open questions in \Cref{sec:open}.

\section{Preliminaries}\label{sec:prelim}

All graphs considered in this paper are finite, simple, and undirected.
Let $G$ be a graph.
For a set $X\subseteq V(G)$, we denote by $N_G(X)$, or simply by $N(X)$ if the graph is clear from the context, the set of vertices in $V(G)\setminus X$ that are adjacent to at least one vertex in $X$.
For $v\in V(G)$, we write $N(v)$ for the set $N(\{v\})$ and refer to its cardinality as the \emph{degree} of $v$; furthermore, we denote by $N[v]$ the set $N(v)\cup \{v\}$.
An \emph{independent set} in $G$ is a set of pairwise nonadjacent vertices.
The \emph{independence number} of $G$, denoted by $\alpha(G)$, is defined as the maximum cardinality of an independent set.
A \emph{matching} in $G$ is a set of pairwise disjoint edges.
An \emph{induced matching} in $G$ is a matching $M$ such that no two endpoints of distinct edges in $M$ are adjacent in $G$.

A graph $G$ is \emph{bipartite} if its vertex set is the union of two independent sets in $G$.
For a positive integer $t$, we denote by $K_{t,t}$ the \emph{balanced complete bipartite graph} with both parts of size $t$, that is, a graph that admits a partition of its vertex set into two parts of size $t$ such that two distinct vertices are adjacent if and only if they belong to different parts. 
For a graph $G$ and a set $X\subseteq V(G)$, we denote by $G[X]$ the subgraph of $G$ induced by $X$, that is, the graph with vertex set $X$ in which two vertices are adjacent if and only if they are adjacent in $G$.
Given two graphs $H$ and $G$, we say that $G$ is \emph{$H$-free} if no induced subgraph of $G$ is isomorphic to $H$.

A \emph{tree decomposition} of a graph $G$ is a pair $\cT = (T,\beta)$ consisting of a tree $T$ and a function $\beta$ defined on the vertex set of $T$ that assigns to each node $t\in V(T)$ a set $\beta(t)\subseteq V(G)$ called a \emph{bag} such that for each edge $e\in E(G)$, there exists a bag containing both endpoints of $e$, and for each vertex $v\in V(G)$, the set of nodes $t\in V(T)$ such that $v\in \beta(t)$ induces a nonempty subtree $T_v$ of $T$.
A \emph{balanced separator} in a graph $G$ is a set $S\subseteq V(G)$ such that no component of the graph $G-S$ contains more than $|V(G)|/2$ vertices.
For every tree decomposition $(T,\beta)$ of a graph $G$, there exists a bag $\beta(t)$ that is a balanced separator in $G$ (see, e.g., \cite[proof of Lemma 7.19]{MR3380745}).

We now define the two main parameters studied in this paper, using the terminology and notation from~\cite{DMS24a,LMMORS-ESA24,MR4906164}. 
For a tree decomposition $\cT = (T,\beta)$ of a graph $G$, the \emph{independence number} of $\cT$, denoted by $\alpha_G(\cT)$ ---or simply by $\alpha(\cT)$ if the graph is clear from the context---, is defined as the maximum independence number of a subgraph of $G$ induced by a bag, that is, $\alpha(\cT) = \max\{\alpha(G[\beta(t)])\colon t\in V(T)\}$.
The \emph{tree-independence number} of a graph $G$, denoted by $\treealpha(G)$, is the minimum independence number of a tree decomposition of $G$.

Similarly, for a tree decomposition $\cT = (T,\beta)$ of a graph $G$, the \emph{induced matching number} of $\cT$, denoted by $\mu_G(\cT)$ ---or simply by $\mu(\cT)$ if the graph is clear from the context---, is defined as the maximum integer $k$ such that $G$ admits an induced matching $M$ with $k$ edges such that some bag of $\cT$ contains at least one endpoint of each edge in $M$, that is, $\mu(\cT) = \max\{\mu(G,\beta(t))\colon t\in V(T)\}$, where, for a set $X\subseteq V(G)$, we denote by $\mu(G,X)$ the maximum cardinality of an induced matching $M$ in $G$ such that $e\cap X\neq \emptyset$ for all $e\in M$. 
The \emph{induced matching treewidth} of a graph $G$, denoted by $\yolov(G)$, is the minimum induced matching number of a tree decomposition of $G$.

\section{Auxiliary Ramsey-type lemmas}\label{sec:auxiliary-lemmas}

In this section we show two auxiliary Ramsey-type results with polynomial bounds.
In both we will use the celebrated K\"ov\'ari-S\'os-Tur\'an theorem~\cite{MR65617}.

\begin{theorem}[K\"ov\'ari-S\'os-Tur\'an]\label{thm:kst}
        For every fixed $t \geq 1$, every $n$-vertex graph that does not contain $K_{t,t}$ as a subgraph has at most $\frac{(t-1)^{1/t}}{2} n^{2-1/t} + \frac{tn}{2}$ edges.
\end{theorem}
This immediately yields the following.
\begin{corollary}\label{cor:kst}
    For every $t \geq 1$ there exists $n_t$ such that for every $n \geq n_t$, every $n$-vertex graph that does not contain $K_{t,t}$ as a subgraph has at most $n^{2-1/t}$ edges.
\end{corollary}

Let us also recall a variant of the Tur\'an's theorem, see, e.g., \cite[pp. 95-96]{AS} and the proof of Theorem~2 therein.

\begin{theorem}[Tur\'an]\label{lem:extractis}
    Let $\sigma \geq 1$ be a real number and let $Q$ be an $n$-vertex graph with at most $\sigma n$ edges.
    Then $Q$ has an independent set of size at least $\frac{n}{2\sigma+1} \geq \frac{n}{3\sigma}$.
\end{theorem}

\paragraph{Extracting a large induced matching.} First, we show that in a $K_{t,t}$-free bipartite graph, from a large matching one can extract an induced matching of polynomial size.

\begin{lemma}\label{lem:extract-induced-matching}
There exists  a function $\M(s,t) = \Oh_t(s^t)$ for which the following holds.
Every bipartite graph that contains a matching of size at least $\M(s,t)$ contains either an induced\/ $K_{t,t}$ or an induced matching of size\/ $s+1$.
\end{lemma}
\begin{proof}
    For fixed $t$, define
    \[\M(s,t) = \max \left( n_t, \lceil(12(s+1))^{t}\rceil\, \right),\]
    where $n_t$ is given by \cref{cor:kst}. Note that  $\M(s,t) =\Oh_t(s^t)$.
    Let $\widetilde{G}$ be a bipartite graph and let $M$ be a matching in $\widetilde{G}$ of size $n \geq \M(s,t)$.
    Suppose that $\widetilde{G}$ is $K_{t,t}$-free; we aim to exhibit an induced matching of size $s+1$.
    Let $G$ be the subgraph of $\widetilde{G}$ induced by the vertices that belong to the edges of $M$.
    Note that $|V(G)|=2n \geq n_t$. Let $m$ be the number of edges of $G$.
    By \cref{cor:kst}, we have $m \leq (2n)^{2-1/t} < 4 \cdot n^{2-1/t}$.

    Let $Q$ be the graph obtained from $G$ by contracting each edge of $M$ (we do not create parallel edges nor loops);
    clearly $|V(Q)| = n$.
    Observe that $|E(Q)| \leq |E(G)| =m$, as every edge of $G$ gives rise to at most one edge of $Q$.
    Applying \cref{lem:extractis} to $Q$ and $\sigma = 4 \cdot n^{1-1/t}$, we obtain an independent set in $Q$ of size at least
    \[
        \frac{n}{3\sigma} = \frac{n}{12 \cdot  n^{1-1/t}} = \frac{n^{1/t}}{12} \geq s+1.
    \]
    Note that an independent set in $Q$ corresponds to an induced matching in $G$ and thus in $\widetilde{G}$.
    This completes the proof.
\end{proof}

\paragraph{Extracting independent sets.} Now, let us show that in a $K_{t,t}$-free graph, given a family of $m$ large independent sets, one can extract from each a subset of size $s$, such that the union of extracted sets is independent. The crux here is that ``large'' is polynomial in both $s$ and $m$.

\begin{lemma}\label{lem:extract-independent-set}
There exists a function $\N(s,t,m) = \Oh_t( (sm^2)^t)$ for which the following holds.
Let\/ $G$ be a\/ $K_{t,t}$-free graph, and let\/ $I_1,\ldots,I_m$ be independent sets in\/ $G$ each of size at least\/ $\N(s,t,m)$.
Then there is an independent set\/ $I$ in\/ $G$ such that\/ $|I\cap I_i|\geq s$ for all\/ $i\in[m]$.
\end{lemma}
 
\begin{proof}
    Fix $t$ and define 
    \[\N(s,t,m) = \max \left( n_t, \lceil(8 sm(m-1))^t\rceil \right),\] where $n_t$ is the constant from \cref{cor:kst}.
    Note that $\N(s,t,m) = \Oh_t( (sm^2)^t )$.

    Let $\widetilde{G}$ be a $K_{t,t}$-free graph, and let $I_1,\ldots,I_m$ be independent sets in $\widetilde{G}$, each of size at least\/ $\N(s,t,m)$.
    By possibly removing some elements from these sets, we may assume that the size of each set is equal to $n$, where $n = \N(s,t,m)$.
    Aiming for a contradiction, suppose that we cannot select subsets as in the statement of the lemma.
    
    Let $G$ be the subgraph of $\widetilde{G}$ induced by $\bigcup_{i = 1}^m I_i$.
    For each $i \in [m]$, randomly select a subset $X_i$ of $I_i$ of size $2s$, uniformly and independently.
    The expected number of edges in $G[X_1 \cup \ldots \cup X_m]$ is $|E(G)| \frac{4s^2}{n^2}$, since any edge $e=uv\in E(G)$ belongs to this induced subgraph with probability $4s^2/n^2$: indeed, if we denote by $i,j\in [m]$ the two indices such that $u\in I_i$ and $v\in I_j$, then the probability of each of the two independent events $u\in X_i$ and $v\in X_j$ equals $2s/n$.

    If there is a choice in which the number of edges is at most $s$, then by removing one endpoint of each such edge we get sets $U_i \subseteq I_i$, each of size at least $2s-s=s$, such that their union is independent.
    Thus we may assume that this is not the case, implying that the random variable counting the number of these edges is always at least $s+1$ and thus 
    \[
        |E(G)| \frac{4s^2}{n^2} \geq s+1 > s,
    \]
    showing that $|E(G)| > \frac{n^2}{4s}$.
    By averaging this implies that there are $1 \leq i < j \leq m$ such that the graph $G_{i,j} := G[I_i \cup I_j]$ has more than $\frac{n^2}{2sm(m-1)}$ edges.
    As $|V(G_{i,j})| = 2n \geq n_t$ and $G_{i,j}$ is bipartite and thus it cannot contain a $K_{t,t}$ as a subgraph, by \cref{cor:kst} we obtain that the number of edges in $G_{i,j}$ is at most $(2n)^{2-1/t}<4 n^{2-1/t}$.
    Consequently,
    \[
    \frac{n^2}{2sm(m-1)}<4 n^{2-1/t}\,,
    \]
    which implies that $n<(8 sm(m-1))^t \le \N(s,t,m)$.
    This contradiction completes the proof.    
\end{proof}

\section{Proof of \cref{thm:biclique}}\label{sec:proof}

We now prove \cref{thm:biclique}, which we restate for convenience.

\thmbiclique*

\begin{proof}
The proof follows the same strategy as the proof of \cite[Theorem 1.1]{MR4906164}, with two lemmas derived from Ramsey's theorem, namely, \cite[Lemma 3.1]{MR4906164} and \cite[Lemma 3.2]{MR4906164}, replaced with \cref{lem:extract-induced-matching,lem:extract-independent-set}, respectively.

Let $\M(\cdot,\cdot)$ and $\N(\cdot,\cdot,\cdot)$ be as claimed in \cref{lem:extract-induced-matching,lem:extract-independent-set}, respectively.
Let
\begin{align*}
\C(\mu,t)&=\N(\M(\mu,t),\:t,\:\M(\mu,t))\text{,} \\
\K(\mu,t)&=2\cdot\M(\mu,t)+\mu\cdot\C(\mu,t)\text{.}
\end{align*}
Since $\M(s,t) = \Oh_t(s^t)$ and $\N(s,t,m) = \Oh_t( (sm^2)^t)$, there exist functions $f(t)$ and $g(t)$ such that $\M(s,t)\le f(t)s^t$ and $\N(s,t,m) \le g(t)(sm^2)^t$ for all positive integer arguments $s,{}t$ (and $s,{}t,{}m$, respectively).
The functions $\N$ and $\M$ are nondecreasing in each of the arguments and therefore
\[
\C(\mu,t)=\N(\M(\mu,t),\:t,\:\M(\mu,t)) \le g(t)\cdot (\M(\mu,t))^{3t}\le  g(t)\cdot (f(t))^{3t}\cdot \mu^{3t^2}
\]
and consequently
\[
\K(\mu,t)=2\cdot\M(\mu,t)+\mu\cdot\C(\mu,t)
\le 2\cdot f(t)\mu^t+ g(t)\cdot (f(t))^{3t}\cdot \mu^{3t^2+1}  = \Oh_t( \mu^{3t^2+1})\text{.}
\]
The rest of the proof is the same as the corresponding part of the proof of~\cite[Theorem 1.1]{MR4906164}.
For completeness, we explain the main steps of the proof, but do not reproduce the proofs of the claims below since all the claims are from~\cite{MR4906164}.

Let $G$ be a $K_{t,t}$-free graph, and let $\cT=(T,\beta)$ be a tree decomposition of $G$ with $\mu(\cT)\leq\mu$.
We aim to show that $\treealpha(G)<\K(\mu,t)$.
Let $S$ be a maximum independent set in $G$.

\begin{claim}[Claim 3.3 in \cite{MR4906164}]
\label{clm:max-ind-set-helper}
For every node\/ $x$ of\/ $T$, it holds that\/ $\alpha(\beta(x)\setminus S)<\M(\mu,t)$.
\end{claim}

Let a vertex $v$ of $G$ be called \emph{light} if $\alpha(N(v))<\C(\mu,t)$ and \emph{heavy} otherwise.
Let $S_\ell$ and $S_h$ be the sets of light  and heavy vertices in $S$, respectively.

\begin{claim}[Claim 3.4 in \cite{MR4906164}]
\label{clm:light}
For every node\/ $x$ of\/ $T$, it holds that\/ $\alpha(N(\beta(x)\cap S_\ell))<\mu\cdot\C(\mu,t)$.
\end{claim}

\begin{claim}[Claim 3.5 in \cite{MR4906164}]
\label{clm:heavy}
For every node\/ $x$ of\/ $T$, it holds that\/ $\size{\beta(x)\cap S_h}<\M(\mu,t)$.
\end{claim}

Recall that for a vertex $v$ of $G$, the subgraph of $T$ induced by the nodes that contain $v$ in their bags is denoted by $T_v$; since $\cT$ is a tree decomposition, $T_v$ is a nonempty tree.
We now construct a tree decomposition $\cT'=(T',\beta')$ of $G$ as follows.
\begin{itemize}
\item The tree $T'$ is obtained from $T$ by adding, for every $s\in S_\ell$, a new leaf node $y_s$ adjacent to some node $x_s$ of $T_s$.
\item For every node $x$ of $T$, we set $\beta'(x)=(\beta(x) \setminus S_\ell)\cup N(\beta(x)\cap S_\ell)$.
\item For every vertex $s\in S_\ell$, we set $\beta'(y_s)=N[s]$.
\end{itemize}

\begin{claim}[Claim 3.6 in \cite{MR4906164}]
\label{clm:treedec}
$\cT'$ is a tree decomposition of\/ $G$.
\end{claim}

The final claim establishes a bound on the independence number of $\cT'$ and is proved using \cref{clm:max-ind-set-helper,clm:light,clm:heavy}.

\begin{claim}[Claim 3.7 in \cite{MR4906164}]\label{clm:treealpha}
$\alpha(\cT')<\K(\mu,t)$.
\end{claim}

Now, the theorem follows directly from \cref{clm:treedec,clm:treealpha}.
\end{proof}

\section{Conclusion and open problems}\label{sec:open}

An obvious way to improve our \cref{thm:biclique} would be to show that if $G$ is $K_{t,t}$-free, then $\treealpha(G)$ is bounded by a function that is polynomial in \emph{both} $\yolov(G)$ and $t$. However, this is not the case, as shown in the next lemma.
Here, by $tK_2$ we mean an induced matching with $t$ edges. Clearly, if $G$ is $tK_2$-free, then $\yolov(G) < t$.

\begin{lemma}\label{lem:nopoly}
For any positive integer $t$ there exists a graph $G$ that is $K_{t,t}$-free and $tK_2$-free, but $\treealpha(G) = 2^{\Omega(t)}$.
\end{lemma}
\begin{proof}
    Let $n = \lfloor 2^{t/3} \rfloor$.
    Let $G$ be a random bipartite graph with sides $A$ and $B$ of size $n$ each, and edges between sides added uniformly at random, that is, independently of each other with probability $p = 1/2$. 
    By the union bound the probability that $G$ contains a copy of $K_{t,t}$ is at most
    $$
    {n \choose t}^2 2^{-t^2} \leq 2^{-t^2/3},
    $$
which tends to $0$ as $t$ tends to infinity. The same argument applied to the bipartite complement of $G$ shows that the probability that there are subsets $X \subset A, Y \subset B$ satisfying $|X|=|Y|=t$  with no edge between them is equally tiny. 
Similarly, the probability that $G$ contains an induced matching of $t$ edges is at most
$$
{n \choose t}^2 t! 2^{-t^2} \leq 2^{t \log t -t^2/3},
$$
which also tends to $0$ as $t$ tends to infinity. Therefore, with high probability, that is, with probability that tends to $1$ as $t$ tends to infinity, $G$ is $K_{t,t}$-free and $tK_2$-free, and in addition contains at least one edge between any two subsets $X \subset A, Y \subset B$, where $|X|=|Y|=t.$ 

Fix a graph $G$ satisfying these three properties. 
The third property implies that if we remove from $G$ a set $S$ of fewer than $n-2t$ vertices, then it is impossible to partition the remaining vertices into two disjoint sets $W$ and $Z$, each of size at least $2t$, with no edge between them. 
Indeed, in such a partition the number of vertices of $A$ in the union $W \cup Z$ is larger than $2t$ and so is the number of vertices of $B$ in this union. 
If $W$ contains at least $t$ vertices of $A$ and at least $t$ vertices of $B$ then there must be an edge between $Z$ and $W$, since $Z$ contains either at least $t$ vertices of $A$ or at least $t$ vertices of $B$. 
Therefore we may assume, without loss of generality, that $|W \cap A| < t$. 
In this case $|W \cap B| >t$, hence, if $|Z \cap A| \geq t$, there is an edge between $Z$ and $W$. 
This implies that $|Z \cap A|<t$, but then the union $W \cup Z$ contains less than $2t$ vertices of $A$, which is impossible. 

By the above argument, there is no balanced separator in $G$ of size less than $n-2t$, implying that in any tree decomposition of $G$ there is a bag of size at least $n-2t$. 
Since $G$ is bipartite, the independence number of the subgraph of $G$ induced by this bag is at least $(n-2t)/2$, which is linear in $n$. 
This completes the proof of the lemma.
\end{proof}

Using similar arguments, it can be shown that the smallest positive integer $\M(s,t)$ satisfying the conclusion of \Cref{lem:extract-induced-matching} is not bounded from above by any polynomial of both $s$ and $t$, and that the smallest integer $\N(s,s,2)$ as in \Cref{lem:extract-independent-set} is exponential in $s$.

On the other hand, \cref{lem:nopoly} does not rule out the possibility of upper-bounding the induced matching treewidth on $K_{t,t}$-free graphs by a function that is polynomial in $t$ but exponential in $\treealpha(G)$.

\begin{question}\label{question2}
Is is true that for every positive integer\/ $\mu$ there exists a polynomial $\mathsf{p}_\mu$ such that every\/ $K_{t,t}$-free graph\/ $G$ with\/ $\yolov(G)\leq\mu$ satisfies\/ $\treealpha(G) \le \mathsf{p}_\mu(t)$?
\end{question}

For a graph $G$, the \emph{induced biclique number} of $G$ is the largest nonnegative integer $t$ such that $G$ contains an induced subgraph isomorphic to $K_{t,t}$.
This lower bound on tree-independence number was recently studied in several contexts (see~\cite{MR4901497,galby2024polynomialtimeapproximationschemesinduced,hilaire2025treewidthversuscliquenumber,dallard2024treewidthversuscliquenumber,MR4800641}).
Using this terminology, \Cref{question2} can be equivalently stated as follows.

\begin{question}
Is it true that for classes of  graphs with bounded induced matching treewidth, tree-independence number is bounded from above by a polynomial function of the induced biclique number?
\end{question}

Finally, let us mention that another result proved by Abrishami et al.~\cite{MR4906164} is that any class of graphs of bounded induced matching treewidth is $\chi$-bounded, i.e., the chromatic number can be upper-bounded in terms of the clique number.
However,  the bound is at least exponential with respect to the clique number.
A natural question to ask is whether these graphs are polynomially $\chi$-bounded.

\begin{question}\label{question3}
Is is true that for every positive integer\/ $\mu$ there exists a polynomial $\mathsf{p}_\mu$ such that every\/ graph\/ $G$ with\/ $\yolov(G)\leq\mu$ and clique number at most $\omega$ satisfies\/ $\chi(G) \le \mathsf{p}_\mu(\omega)$?
\end{question}

\bibliographystyle{plainurl}
\begin{sloppypar}
\bibliography{biblio}
\end{sloppypar}

\end{document}